\documentclass[10pt,twocolumn]{IEEEtran}

\usepackage{amsthm, amssymb}

\newtheoremstyle{slplain}
  {3pt}
  {3pt}
  {\slshape}
  {}
  {\bfseries}
  {.}%
  { }
  {}

\theoremstyle{slplain}

\newtheorem{thm}{Theorem}
\newtheorem{cor}{Corollary}
\newtheorem{lem}{Lemma}
\newtheorem{pro}{Proposition}

\newtheorem{defi}{Definition}

\usepackage{cite}
\usepackage{amsfonts}
\usepackage{multicol,multienum}
\usepackage{subfigure}

\ifCLASSINFOpdf
  \usepackage[pdftex]{graphicx}
  \graphicspath{{../eps/}{../jpeg/}}
  \DeclareGraphicsExtensions{.eps}
\else
  \usepackage{float}
  \usepackage[dvips]{graphicx}
  \graphicspath{{../eps/}}
  \DeclareGraphicsExtensions{.eps}
\fi

\usepackage[cmex10]{amsmath}

\usepackage{algorithm}
\usepackage{algorithmic}
\usepackage{array}
\usepackage{setspace}
\usepackage{url}

\begin{document}

\title{Towards Understanding the Fundamentals of Mobility in Cellular Networks}

\author{
\IEEEauthorblockA{Xingqin Lin, Radha Krishna Ganti, Philip J. Fleming and Jeffrey G. Andrews}
\thanks{Xingqin Lin and Jeffrey G. Andrews are with Department of Electrical $\&$ Computer Engineering, The University of Texas at Austin, USA. (E-mail: xlin@utexas.edu, jandrews@ece.utexas.edu). Radha Krishna Ganti is with Department of Electrical Engineering, 
Indian Institute of Technology, Madras, India. (E-mail: rganti@ee.iitm.ac.in). Philip J. Fleming is with Nokia Siemens Networks (E-mail: phil.fleming@nsn.com). This research was supported by Nokia Siemens Networks. Date revised: \today}
}

\maketitle
%\doublespacing
\thispagestyle{empty}

\begin{abstract}
Despite the central role of mobility in wireless networks, analytical study on its impact on network performance is notoriously difficult. This paper aims to address this gap by proposing a random waypoint (RWP) mobility model defined on the entire plane and applying it to analyze two key cellular network parameters: handover rate and sojourn time. We first analyze the stochastic properties of the proposed model and compare it to two other models: the classical RWP mobility model and a synthetic truncated Levy walk model which is constructed from real mobility trajectories. The  comparison shows that the proposed RWP mobility model is more appropriate for the mobility simulation in emerging cellular networks, which have ever-smaller cells. Then we apply the proposed model to cellular networks under both deterministic (hexagonal) and random (Poisson) base station (BS) models. We present analytic expressions for both handover rate and sojourn time, which have the expected property that the handover rate is proportional to the square root of BS density. Compared to an actual BS distribution, we find that the Poisson-Voronoi model is about as accurate in terms of mobility evaluation as hexagonal model, though being more pessimistic in that it predicts a higher handover rate and lower sojourn time.
\end{abstract}

\IEEEpeerreviewmaketitle

\section{Introduction}
\label{sec:intro}

The support of mobility is a fundamental aspect of wireless networks \cite{camp2002survey, akyildiz1999mobility}.  Mobility management is taking on new importance and complexity in emerging cellular networks, which have ever-smaller and more irregular cells.

\subsection{Background and Proposed Approach}

As far as cellular networks are concerned, typical questions of interest include how mobility affects handover rate and sojourn time. Handover rate is defined as the expected number of handovers per unit time. It is directly related to the network signaling overhead. Clearly, the handover rate is low for large cells and/or low mobility, but smaller cells are necessary to increase the capacity of cellular networks through increased spectral and spatial reuse. Thus, analytic results on handover rate will be useful for network dimensioning, and in order to understand tradeoffs between optimum cell associations and the undesired overhead in session setup and tear downs. Sojourn time is defined as the time that a mobile resides in a typical cell. It represents the time that a BS would provide service to the mobile user.  In the case of a short sojourn time, it may be preferable from a system-level view to temporarily tolerate a suboptimal BS association versus initiating a handover into and out of this cell.  This concept is supported in 4G standards with a threshold rule \cite{3gppMobility}, but it is fair to say that the theory behind such tradeoffs is not well developed. Though handover rate is inversely proportional to expected sojourn time, their distributions can be significantly different, which motivates us to study them separately in this paper.

To explore the role of mobility in cellular networks, particularly handover rate and sojourn time, mobility modeling is obviously a necessary first step. In this paper, we focus on RWP mobility model originally proposed in \cite{johnson1996dynamic} due to its simplicity in modeling movement patterns of mobile nodes. In this model, mobile users move in a finite domain $\mathcal{A}$. At each turning point, each user selects the destination point (referred to as waypoint)  uniformly distributed in $\mathcal{A}$ and chooses the velocity from a uniform distribution. Then the user moves along the line (whose length is called transition length) connecting its current waypoint to the newly selected waypoint at the chosen velocity. This process repeats at each waypoint. Optionally, the user can have a random pause time at each waypoint before moving to the next waypoint.  In this classical RWP mobility model, the stationary spatial node distribution tends to concentrate near the center of the finite domain and thus may be inconvenient if users locate more or less uniformly in the network \cite{bettstetter2004stochastic}. Another inconvenience is that the transition lengths in the classical RWP mobility model are of the same order as the size of the domain $\mathcal{A}$, which seems to deviate significantly from those observed in human walks \cite{rhee2011levy}. 

To solve the inconveniences mentioned above, we propose a RWP mobility model defined on the entire plane. In this model, at each waypoint the mobile node chooses 1) a random direction uniformly distributed on $[0,2\pi]$, 2) a transition length from some distribution, and 3) a velocity from some distribution. Then the node moves to the next waypoint (determined by choice 1 and 2) at the chosen velocity. As in the classical RWP mobility model, the node can have a random pause time at each waypoint.  Note that human movement has very complex temporal and spatial correlations and its nature has not been fully understood \cite{gonzalez2008understanding, song2010modelling}. It is fair to say that none of the existing mobility models are fully realistic. The motivation of this work is not to solve this open problem. Instead, we aim to propose a  tractable RWP mobility model, which can be more appropriate than the classical RWP mobility model for the study of mobility in emerging cellular networks, so that it can be utilized to analyze the impact of mobility in cellular networks to provide insights on network design. 
% Thus, the proposed RWP mobility model, as the classical RWP mobility model and its variants, can only  model human walks in a simplified and thus non-realistic manner.  

After presenting the proposed mobility model, we analyze its associated stochastic properties and compare it to two other models: the classical RWP mobility model and a synthetic truncated Levy walk model\cite{rhee2011levy} which is constructed from real mobility trajectories. The  comparison shows that the proposed model is more appropriate for the mobility simulation in emerging cellular networks. Then the proposed RWP mobility model is applied to cellular networks whose BSs are modeled in two conceptually opposite ways.  The first is the traditional hexagonal grid, which represents an extreme in terms of regularity and uniformity of coverage, and is completely deterministic.  The second is to model the BS locations as drawn from a Poisson point process (PPP), which creates a set of BSs with completely independent locations \cite{StoFle97,Bro00}.   Under the PPP BS model, the cellular network can be viewed as a Poisson-Voronoi tessellation if the mobile users are assumed to connect to the nearest BSs.  As one would expect, most actual deployments of cellular networks lie between these two models, both qualitatively -- i.e. they are neither perfectly regular nor perfectly random -- and quantitatively, i.e. the SINR statistics and other statistical measures are bounded by these two approaches \cite{andrews2010tractable}. Thus, both models are of interest and we apply the proposed RWP mobility model to cellular networks under both models. Analytic expressions for handover rate and sojourn time are obtained under both models, some of which are quite simple and lead to intuitive interpretations.

In some aspects the proposed RWP mobility model is similar to the random direction (RD) mobility model. Though the classical RD mobility model does not have a non-homogenous spatial distribution as in the classical RWP mobility model, its transition lengths are still of the same order as the size of the simulation area, which seems to deviate significantly from those in human walks \cite{royer2001analysis}. This inconvenience may be solved by a modified RD mobility model \cite{nain2005properties}, in which the transition length distribution is co-determined by the velocity distribution and moving time distribution. This may complicate the theoretical analysis in this paper. This motivates us to propose a mobility model with transition length distribution directly specified. Besides, transition length data seems to be more readily available than moving time data in many real data sets \cite{rhee2011levy, song2010modelling}. Thus, it might be easier to use the proposed RWP mobility model when fitting the mobility model to real mobility trajectories.

Before ending this subsection, it is worth mentioning the many trace-based mobility models \cite{hong1999group, jardosh2003towards, musolesi2004ad, saha2004modeling, hsu2005weighted, mcnett2005access, kim2006extracting, lee2009slaw, hsu2009modeling, aschenbruck2011trace, cho2011friendship}.  One main drawback of trace-based mobility models is that, due to the differences in the trace acquiring methods, sizes of trace data, and data filtration techniques, mobility model built on one trace data set may not be applicable to other network scenarios. Moreover, trace-based mobility models are often not mathematically tractable (at least very complicated), preventing researchers from analytically studying the performance of various protocols and/or getting quick informative results in mobile networks. In contrast, random synthetic models, including random walk, the classical and the proposed RWP, Gauss-Markov, etc., are generic and more mathematically tractable. In addition, some aspects of some random synthetic models (including the proposed RWP model) can be fitted  using traces, as done in \cite{mcnett2005access}. Nevertheless, trace-based mobility models are more realistic and scenario-dependent, which is crucial to perform reliable performance evaluation of mobile networks. Hence, both random and trace-based synthetic models are important for the study of mobile networks. In particular, when evaluating a specific protocol in mobile networks, researchers can utilize random synthetic models for example the proposed RWP model to quickly get informative results. Meanwhile, researchers can build specific scenario-dependent trace-based model by collecting traces if possible. Then based on the constructed trace-based model researchers can further perform more extensive and reliable evaluation of the protocol beyond just using random synthetic models.

\subsection{Related work}

The proposed RWP model is based on the one originally proposed in \cite{johnson1996dynamic}. Due to its simplicity in modeling movement patterns of mobile nodes, the classical RWP mobility model has been extensively studied in literature \cite{bettstetter2004stochastic, hyytia2006spatial,bettstetter2003node, resta2002analysis}. These studies analyzed the various stochastic mobility parameters, including transition length, transition time, direction switch rate, and spatial node distribution. When it comes to applying the mobility model to hexagonal cellular networks, simulations are often required to study the impact of mobility since the analysis is hard to proceed \cite{akyildiz1999mobility, 3gppMobility}. Nonetheless, the effects of the classical RWP mobility model to cellular networks have been briefly analyzed in \cite{bettstetter2003node} and a more detailed study can be found in \cite{hyytia2007random}. However, as remarked above, the classical RWP model may not be convenient in some cases. In contrast, we analyze and obtain insight about the impact of mobility under a hexagonal model through applying the relatively clean characterization of the proposed RWP model, as \cite{hyytia2007random} did through applying the classical RWP model.

The application of the proposed RWP mobility model to cellular networks modeled as Poisson-Voronoi tessellation requires stochastic geometric tools, which are becoming increasingly sophisticated and popular \cite{stoyan1995stochastic, HaenggiBook, haenggi2009stochastic, baccelli2009stochastic}.  As far as mobility is concerned,  \cite{baccelli1997stochastic} proposed a framework to study the impact of mobility in cellular networks modeled as Poisson-Voronoi tessellation. In particular, the authors proposed a Poisson line process to model the road system, along which the mobile users move. Thus, the mobility pattern in  \cite{baccelli1997stochastic} is of large scale while the RWP mobility model in this paper is of small scale. Our study can be viewed as  complementary to \cite{baccelli1997stochastic}. For example, our result indicates that if cells decrease in size such that the BS density per unit area is increased by 4 times, then the handover rate would be doubled. Note that recent work showed that the PPP model for BSs was about as accurate in terms of SINR distribution as the hexagonal grid for a representative urban cellular network \cite{andrews2010tractable}.   Interestingly, we find that this observation is also true for mobility evaluation, though the Poisson-Voronoi model yields slightly higher handover rate and lower sojourn time and thus is a bit more pessimistic than the hexagonal model (which is correspondingly optimistic).

We briefly summarize the contributions of this work: 1) We propose a tractable RWP mobility model which overcomes some inconveniences of the classical RWP mobility model and is more appropriate for mobility study in cellular networks. 2) We obtain analytical results for handover rate and sojourn time in cellular networks. Though a specific transition length distribution is assumed, the analysis in this paper is quite general and can be extended to any other transition length distribution which has finite mean.
3) We connect the mobility results for the two conceptually opposite models of cellular networks: hexagonal and Poisson-Voronoi tessellation.

\section{Proposed RWP mobility model on infinite planes}
\label{sec:proposed}

We describe the proposed RWP mobility model in this section.  As in the description of the classical RWP model  (see, e.g., \cite{bettstetter2004stochastic}), the movement trace of a node  can be formally described by an infinite sequence of quadruples:
$
\{(\boldsymbol{X}_{n-1}, \boldsymbol{X}_{n}, V_n, S_n ) \}_{n \in \mathbb{N}} 
$
, where $n$ denotes the $n$-th movement period. During the $n$-th movement period, $\boldsymbol{X}_{n-1}$ denotes the  starting waypoint,
$\boldsymbol{X}_{n}$ denotes the  target waypoint,
$V_n$ denotes the velocity,
and $S_n$ denotes the pause time  at the waypoint $\boldsymbol{X}_{n}$. Given the current waypoint $\boldsymbol{X}_{n-1}$, the next waypoint $\boldsymbol{X}_{n}$ is chosen such that the included angle between the vector $\boldsymbol{X}_{n} - \boldsymbol{X}_{n-1}$ and the abscissa is uniformly distributed on $[0,2\pi]$ and the transition length $L_n= \parallel \boldsymbol{X}_{n} - \boldsymbol{X}_{n-1} \parallel$ is a nonnegative random variable. The selection of waypoints is independent and identical for each movement period.

Though there is a degree of freedom in modeling the random transition lengths, we focus on a particular choice  in this paper. Specifically, the transition lengths $\{L_1, L_2, ...\}$ are chosen to be independent and identically distributed (i.i.d.) with cumulative distribution function (cdf)
\begin{align}
P ( L \leq l )  = 1 - \exp(- \lambda \pi l^2), l \geq 0,
\end{align}
i.e., the transition lengths are Rayleigh distributed. Velocities $V_n$ are i.i.d. with distribution $P_V(\cdot)$. Pause times $S_n$ are also i.i.d.  with distribution $P_S (\cdot)$. This selection bears an interesting interpretation: Given the waypoint $\boldsymbol{X}_{n-1}$, a homogeneous PPP $\Phi(n)$ with intensity $\lambda$ is independently generated and then the nearest point in $\Phi (n)$ is selected as the next  waypoint, i.e.,
$
\boldsymbol{X}_{n} = \textrm{arg} \min_{\boldsymbol{x} \in \Phi(n)} \parallel \boldsymbol{x} - \boldsymbol{X}_{n-1} \parallel .
$

Under the proposed model, different mobility patterns can be captured by choosing different $\lambda$'s.  Larger $\lambda$ statistically implies that the transition lengths $L$ are shorter. This further implies that the movement direction switch rates are higher. These mobility parameters may be appropriate for mobile users walking and shopping in a city, for example. In contrast, smaller $\lambda$ statistically implies that the transition lengths $L$ are longer and the corresponding movement direction switch rates are lower. These mobility parameters may be appropriate for driving users, particularly those on the highways. This intuitive result can also be observed in Fig. \ref{fig:2}, which shows $4$ sample traces of the proposed RWP model.

\begin{figure}
\centering
\includegraphics[width=8cm]{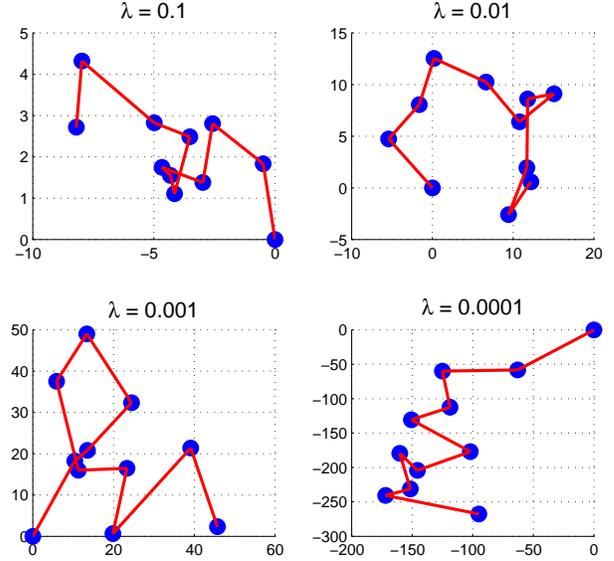}
\caption{Sample traces of the proposed RWP mobility models. The transition lengths are statistically shorter with larger mobility parameter $\lambda$, and vice versa.}
\label{fig:2}
\end{figure}

\section{Stochastic properties of the Proposed RWP mobility model}
\label{sec:stochastic}

In this section, we  first study the various stochastic properties of the proposed RWP mobility model. Stochastic properties of interests include transition length, transition time, direction switch rate, and the spatial node distribution.  We then perform simulation to compare the proposed RWP mobility model to the classical RWP mobility model and a synthetic Levy Walk model proposed in \cite{rhee2011levy}, which is constructed from real mobility trajectories.

\subsection{Transition length}

We define \textit{transition length} as the Euclidean distance between two successive waypoints. In the proposed model, the transition lengths can be described by a stochastic process $\{L_n\}_{n \in \mathbb{N}}$ where $L_n$ are i.i.d. Rayleigh distributed with 
\begin{align}
\mathbb{E} [ L] = \frac{1}{2\sqrt{\lambda}}.
\end{align}
Note that the transition lengths are not i.i.d. in the classical RWP mobility model. Indeed, a node currently located near the border of the finite domain tends to have a longer transition length while a node located around the center of the finite domain statistically has a shorter transition length for the next movement period. As a result, it is difficult to obtain the probability distribution for each random transition length.

Nevertheless, the random waypoints $\boldsymbol{X}_{n}$ are i.i.d. in the classical RWP mobility model, which is obvious since they are selected uniformly from a finite domain and independently over movement periods. This property forms the basis for the analysis of the classical RWP mobility model (see, e.g., \cite{bettstetter2004stochastic, hyytia2006spatial}). In contrast, the waypoints in our proposed model are not i.i.d. but form a Markov process.

\subsection{Transition time}
\label{subsec:1}

We define \textit{transition time} as the time a node spends during the movement between two successive waypoints. We denote by $T_n$ the transition time for the movement period $n$. Then 
$T= L/V $
where we omit the period index $n$ since $T_n$ are i.i.d.. Denote $\mathcal{V} \in \mathbb{R}$ as the range of the random velocity $V$. Given any velocity distribution $P_V(\cdot)$, the probability distribution of $T$ is given as follows.

\begin{pro}
The cdf of the random transition time $T$ is given by
\begin{equation}
P( T \leq t ) = 1 - \int_{\mathcal{V}}  \exp(- \lambda \pi v^2 t^2) \ {\rm d} P_V (v), t\geq 0.
\end{equation}
\label{pro:4}
\end{pro}

The proof of Prop. \ref{pro:4} is omitted for brevity. As a specific application of Prop. \ref{pro:4}, the following corollary gives closed form expressions for transition times under two types of velocity distributions.

\begin{cor}
\begin{enumerate}
\item 
If $V \equiv \nu$ where $\nu$ is a positive constant, the pdf of transition time $T$ is 
\begin{align}
f_{T} (t) = 2\pi \lambda \nu^2 t  e^{- \lambda \pi \nu^2 t^2}, t\geq 0.
\label{eq:11}
\end{align}
\item
If $V$ is uniformly distributed on $[v_{\min}, v_{\max}]$, the pdf of transition time $T$ is 
\begin{align}
f_{T} (t) = \frac{g(v_{\min}) - g(v_{\max})}{(v_{\max} - v_{\min})t}, t \geq 0,
\label{eq:12}
\end{align}
where $g(x) \triangleq x e^{- \lambda \pi t^2 x^2} + \displaystyle \frac{1}{\lambda^{1/2} t} Q(\sqrt{2\pi \lambda} t x) $ is non-increasing, and $\displaystyle Q(x) = \frac{1}{\sqrt{2 \pi}} \int_{x}^{\infty} e^{-\frac{u^2}{2}} \ {\rm d} u $.
\end{enumerate}
\label{cor:2}
\end{cor}
Next, we derive the mean transition time. Instead of applying the cdf of $T$, we notice that
\begin{align}
E[T] &= E [\frac{L}{V}] = E[L] E [\frac{1}{V}] \notag \\
&= E[L] \int_{\mathcal{V}} \frac{1}{v} \ {\rm d} P_V (v)  = \frac{1}{2 \sqrt{\lambda}} \int_{\mathcal{V}} \frac{1}{v} \ {\rm d} P_V (v).
\label{eq:14}
\end{align}   
From (\ref{eq:14}), we can easily obtain that, if $V \equiv \nu$,
\begin{align}
E[T] = \frac{1}{2 \nu \sqrt{\lambda}}, 
\label{eq:15}
\end{align}
and, if $V$ is uniformly distributed on $[v_{\min}, v_{\max}]$,
\begin{align}
E[T] = \frac{\ln v_{\max} - \ln v_{\min}}{2 \sqrt{\lambda} (v_{\max} - v_{\min}) }.
\label{eq:16}
\end{align}
From ($\ref{eq:16}$), it is clear that $v_{\min} > 0$ is required to ensure finite expected transition time.

\subsection{Direction switch rate}

In this subsection, we study the \textit{direction switch rate}, which is the inverse of the time between two direction changes and thus characterizes the frequency of direction change and is also a mobility parameter of interest \cite{bettstetter2004stochastic}. To this end, we first introduce the \textit{period time}, the time a node spends between two successive waypoints. In particular, period time $T_p$ is the sum of the pause time and the transition time, i.e., $T_p = T + S$, where $S$ denotes the random pause time. Then the direction switch rate denoted by $D$ is defined to be the inverse of the period time:
$D = 1/T_p$,
whose probability distribution is described in the following proposition.
\begin{pro}
The cdf of the direction switch rate $D$ is given by
\begin{align}
P( D \leq d ) = \int_{\mathcal{S}}  \int_{\mathcal{V} }    \exp \left(- \lambda \pi v^2 \frac{(d-s)^2}{d^2}  \right)   \ {\rm d} P_V (v) \ {\rm d} P_S (s),
\label{eq:18}
\end{align}
where $\mathcal{S} \in \mathbb{R}$ denotes the range of the random pause time.
\label{pro:5}
\end{pro}

We omit the proof of Prop. \ref{pro:5} for brevity. Given the distributions of velocity and transition time, i.e., $P_V (v)$ and $P_S (s)$, the distribution of direction switch rate $D$ can be found by Prop. \ref{pro:5}. 

\subsection{Spatial node distribution}

In this subsection, we study the \textit{spatial node distribution}. To this end, let $\boldsymbol{X}_0$ and $\boldsymbol{X}_1$ be two successive waypoints. Given $\boldsymbol{X}_0$, we are interested in the probability that the moving node resides in some measurable set $\mathcal{A}$ during the movement from $\boldsymbol{X}_0$ to $\boldsymbol{X}_1$. We first derive the spatial node distribution given in the following theorem with the assumption that the mobile node does not have pause time.

\begin{thm}
Assume that $S_n \equiv 0, \forall n$, and that   $\boldsymbol{X}_0$ is at the origin. Then the spatial node distribution between $\boldsymbol{X}_0$ and $\boldsymbol{X}_1$ is characterized by the pdf $f(r, \theta)$ given by
\begin{equation}
f(r, \theta) = \frac{\sqrt{\lambda}}{\pi r} \exp (-\lambda \pi r^2 ).
\label{eq:20}
\end{equation}
\label{pro:6}
\end{thm}
\begin{proof}
See Appendix \ref{proof:pro6}.
\end{proof}

The physical interpretation of $f(r, \theta)$ is as follows. Let ${\rm d}A(r,\theta)$ be a small area around the point $(r,\theta)$  given in polar coordinate. Then the probability $P({\rm d}A(r,\theta))$ that the moving node resides in some measurable set $\mathcal{A}$ during the movement from $\boldsymbol{X}_0$ to $\boldsymbol{X}_1$ is approximately given by
\begin{align}
P({\rm d}A(r,\theta))  \approx  f(r, \theta) \cdot |{\rm d}A(r,\theta) | ,
\label{eq:21}
\end{align}
where $| {\rm d}A(r,\theta) |$ denotes the area of the set ${\rm d}A(r,\theta)$.
Also, as noted in the proof in Appendix \ref{proof:pro6}, $f(r, \theta)$ can be regarded as the ratio of the expected proportion of transition time in the set ${\rm d}A(r,\theta)$ to $| {\rm d}A(r,\theta) |$.

Now let us consider the case where the mobile node has random pause time $S$, characterized by the probability measure $P_S (\cdot)$. In this case, the spatial node distribution is given in the following theorem.
\begin{thm}
Assume that $\boldsymbol{X}_0$ is at the origin. Then the spatial node distribution between $\boldsymbol{X}_0$ and $\boldsymbol{X}_1$ is characterized by the pdf $\tilde{f}(r, \theta)$:
\begin{align}
\tilde{f}(r, \theta) = p \cdot \frac{\sqrt{\lambda}}{\pi r} \exp (-\lambda \pi r^2 ) + (1-p) \cdot \lambda \exp(-\lambda \pi r^2),
\label{eq:22}
\end{align}
where $\displaystyle p = \frac{E[T]}{E[T] + E[S]}$ is the expected proportion of transition time from $\boldsymbol{X}_0$ to $\boldsymbol{X}_1$.
\label{thm:2}
\end{thm}
\begin{proof}
Note that the pause time $S$ is independent of both $\boldsymbol{X}$ and $V$. Besides, the non-static probability $p$ is the expected proportion of the time that the node $i$  is moving, i.e.,
\begin{align}
p = \lim_{n \to \infty} \frac{ \sum_{m=1}^{n} t_m }{ \sum_{m=1}^{n} ( t_m + s_m ) } = \frac{E[T]}{E[T] + E[S]}. 
\label{eq:athm22}
\end{align}
The pdf $\tilde{f}(r,\theta)$ is the weighted superposition of two independent components as \cite{resta2002analysis}:
\begin{align}
\tilde{f}(r,\theta) = p  \cdot  f(r,\theta) + (1-p) \cdot f_{\boldsymbol{X}_1}(r,\theta),
\label{eq:athm21}
\end{align}
where $f(r,\theta)$ is the spatial node distribution without pause time and is given in Theorem \ref{pro:6} , and $f_{\boldsymbol{X}_1}(r,\theta)$ denotes the spatial distribution of the waypoint $\boldsymbol{X}_1$ and is given in Lemma \ref{lem:1}. Plugging $f(r,\theta)$ and $f_{\boldsymbol{X}_1}(r,\theta)$ into (\ref{eq:athm21}), the expression (\ref{eq:22}) follows.
\end{proof}

\subsection{Model comparison}

In this section, we perform simulations to study the difference between the proposed RWP mobility model and the classical one. For validation, we also compare them to a synthetic truncated Levy walk model \cite{rhee2011levy}, which is constructed from real mobility trajectories. Thus, we indirectly compare the two RWP mobility models to real human walks. In the truncated Levy walk model, transition lengths have an inverse power-law distribution: $P_L(l) \sim \frac{1}{l^{1+\alpha}}, 0<\alpha<2$. The pause times also have an inverse power-law distribution: $P_S(s) \sim \frac{1}{s^{1+\beta}}, 0<\beta<2$.

We simulate the movement of a mobile node under the three mobility models, respectively. The simulation area is a 1000$\times$1000 grid.\footnote{Note that, though the proposed RWP mobility model is defined on the entire plane for analytical tractability, in a simulation one has to deal with the boundary issue. Here we use reflection method.} Since the main assumption in the proposed model is the Rayleigh distributed transition lengths, we first compare the statistics of transition lengths under the three mobility models. The results are shown in Fig. \ref{fig:22}, where $\lambda$ is chosen such that the proposed RWP mobility model has the same mean transition length as that of Levy walk model. As shown in Fig. \ref{fig:22}, transition lengths of the classical RWP mobility model are statistically much larger than those of Levy walk model. In contrast, transition lengths of the proposed RWP mobility model match those of Levy walk model better, especially in the low transition length regime.  However, the proposed RWP mobility model lacks the heavy tail characteristic of the Levy walk model. This mismatch in the high transition length regime gradually diminishes as $\alpha$ increases.

\begin{figure}
\centering
\includegraphics[width=8cm]{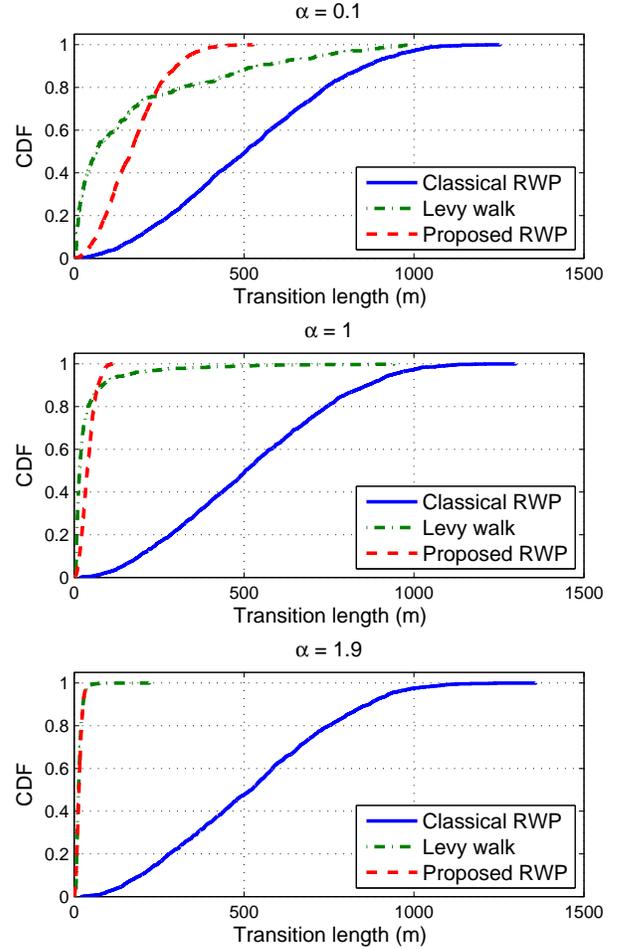}
\caption{Simulated statistics of transition lengths: $V \equiv 1$ and $\beta=1$.}
\label{fig:22}
\end{figure}

We next compare the statistics of direction switch rate under the 3 mobility models. The results are shown in Fig. \ref{fig:222}, where $\lambda$ is chosen such that the proposed RWP mobility model has the same mean transition length as that of Levy walk model. The pause time distribution used in the 2 RWP mobility models are the same power-law distribution as that used for pause time distribution in the Levy walk model. As expected, the classical RWP mobility model has much lower direction switch rate in all the 3 subplots than the Levy walk model. In contrast, the difference in direction switch rate between the proposed RWP mobility model and Levy walk model is moderate when $\beta$ is large. When $\beta$ is small, e.g., $\beta=0.1$, the statistics of the direction switch rate of the proposed RWP mobility model are almost indistinguishable from those of the Levy walk model.

\begin{figure}
\centering
\includegraphics[width=8cm]{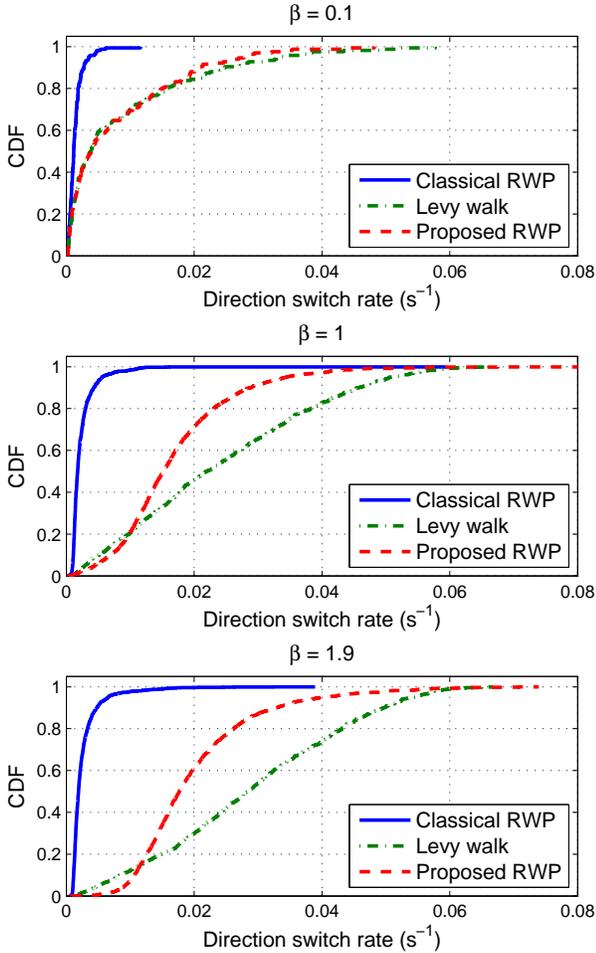}
\caption{Simulated statistics of direction switch rate: $V \equiv 1$ and $\alpha=1$.}
\label{fig:222}
\end{figure}

The above comparison shows that the proposed RWP model is more flexible for mobility simulation in future cellular networks which have ever-smaller cells than the classical RWP model. In particular, some aspects of the proposed RWP model can be directly or indirectly fitted using traces, while the classical RWP model is fixed. Nevertheless, the comparison does not imply that the statistical difference between the proposed RWP and Levy walk model are non-significant. Also, we do not claim that the proposed RWP model can fully represent human walks. Also, human walks are not Levy walks even though the Levy walk model is constructed from real mobility trajectories. Indeed, human walks have complex temporal and spatial correlations and its nature has not yet been fully understood \cite{gonzalez2008understanding}. It is fair to say that none of the existing mobility models can fully represent human walks.  

Note that though many real mobility trajectories show that transition lengths of human walks seem to have an inverse power-law distribution,   we do not choose it to model the transition lengths in our proposed RWP mobility model. The reason is that the mean transition length under  inverse power-law distribution is infinite. This will cause analytical problems, e.g., the expected number of handovers in a movement period becomes infinite. Nevertheless, the analysis in this paper can be readily extended to any other transition length distribution which has finite mean.

\section{Applications to Hexagonal Modeled Cellular Networks}
\label{sec:app1}

In this and the next section, we study the impact of mobility on important cellular network parameters, particularly handover rate and sojourn time, using the proposed RWP mobility model. We focus on hexagonal cellular networks in this section.

\subsection{Handover Rate}
\label{subsec:HH}

We assume the typical mobile user is located at the origin. Then the expected number of handovers during one movement period can be computed as 
\begin{align}
E[N]= \sum_{n = 1}^{\infty} n \int_{C_n} P ({\rm d}A(r,\theta)),
\label{eq:24}
\end{align}
where $P(dA(r,\theta))$ is the probability distribution of the waypoint density $f_{\boldsymbol{X}_1}(r,\theta)$ given in Lemma \ref{lem:1} and $C_n$ denotes the area covered by the $n$-th layer neighbouring cells. Now we formally define the handover rate.
\begin{defi}
The handover rate is defined as the expected number $E[N]$ of handovers during one movement period divided by the expected period time. Mathematically, handover rate is given by  $H = E[N] / E[T_p]$.
\label{defi:1}
\end{defi}

Note that $E[T_p]=E[T]+E[S]$ where $E[T]$ has been given in Prop. \ref{pro:4} and $E[S]$ can be determined from the pause time distribution. However, exact computation of $N$  by (\ref{eq:24}) is tedious. Thus, we propose the following approximation formula 
 \begin{align}
E[N]_{\textrm{app}} = \sum_{n = 1}^{\infty} n \int_{0}^{2\pi} \int_{(2n-1) R}^{(2n+1)R} f_{\boldsymbol{X}_1}(r,\theta) r  {\rm d}r  {\rm d}\theta,
 \label{eq:26}
 \end{align}
 where $R = \sqrt { C/\pi } $ with  $C$ being the hexagonal cell size. In other words, we approximate the $n$-th neighbouring layer by a ring with inner radius $(2n-1) R$ and outer radius $(2n+1)R$. This approximation captures the essence of (\ref{eq:24}) and allows us to derive closed form results on $E[N]$ and the corresponding lower and upper bounds.
\begin{pro}
Let $d$ be the side length of the hexagonal cell and $\lambda$ the mobility parameter. The approximation of the expected number of handovers during one movement period is given by
\begin{align}
E[N]_{\textrm{app}} = \sum_{n = 0}^{\infty} \exp \left( - \frac{3\sqrt{3}}{2} (2n+1)^2  \lambda  d^2 \right),
\label{eq:27}
\end{align}
and is bounded as $E[N]^L_{\textrm{app}} \leq E[N]_{\textrm{app}}  \leq  E[N]^U_{\textrm{app}}$ where
\begin{align}
E[N]^L_{\textrm{app}}  &\triangleq  \sqrt{\frac{\pi}{6\sqrt{3} \lambda d^2}} Q \left( \sqrt{3 \sqrt{3} \lambda d^2 } \right),  \\
E[N]^U_{\textrm{app}} &\triangleq \sqrt{\frac{\pi}{6\sqrt{3} \lambda d^2}} \left ( 1 - Q \left( \sqrt{3 \sqrt{3} \lambda d^2 } \right) \right ). 
\end{align}
Moreover, the difference $\triangle N_{\textrm{app}} (\lambda d^2)$ between the upper bound and lower bound is a strictly increasing function of $\lambda d^2$ and is within the range $(0,1)$. In particular, $\triangle N_{\textrm{app}} (\lambda d^2) \to 0 \textrm{ as } \lambda d^2 \to \infty$, and $\triangle N_{\textrm{app}} (\lambda d^2) \to 1 \textrm{ as } \lambda d^2 \to 0$.
\label{pro:7}
\end{pro} 
\begin{proof}
See Appendix \ref{proof:lem2}.
\end{proof}

Using Prop. \ref{pro:7}, the approximate handover rate can be computed as $H_{\textrm{app}} = \frac{1}{E[T_p]} \cdot E[N]_{\textrm{app}} $ and is bounded as $H^L_{\textrm{app}} \leq H_{\textrm{app}}  \leq  H^U_{\textrm{app}}$, where  $H^L_{\textrm{app}} = \frac{1}{E[T_p]} \cdot E[N]^L_{\textrm{app}},  H^U_{\textrm{app}} = \frac{1}{E[T_p]} \cdot E[N]^U_{\textrm{app}}  $.
As the size of the cells in cellular networks becomes smaller, higher handover rates are expected. In this regard, it is interesting to examine the asymptotic property of handover rate as  in Corollary \ref{cro:2}. 
\begin{cor}
 Assume that any of the following asymptotic conditions holds: 1) $d \to 0$ with fixed $\lambda$; 2) $\lambda \to 0$ with fixed $d$; 3) $\lambda d^2 \to 0$. Then the asymptotic approximate handover rate is given by
 \begin{equation}
 \setlength{\belowdisplayskip}{-15pt}
 {H}_{\textrm{app}} \sim \frac{1}{E[T]+E[S]} \sqrt{\frac{\pi}{6\sqrt{3}\lambda}} \frac{1}{2d}.
 \end{equation}
\label{cro:2}
\end{cor}

Though derivation by (\ref{eq:24}) is not tractable, the \textit{exact} handover rate in hexagonal model can be obtained using a generalized solution of Buffon's needle \cite{schroeder1974buffon} as in the following Proposition.
\begin{pro}
Let $d$ be the side length of the hexagonal cell and $\lambda$ the mobility parameter. 
The expected number of handovers $E[N]$ during one movement period  is given by
\begin{equation}
\setlength{\belowdisplayskip}{3pt}
E [ N ] = E[T] \cdot \frac{4 \sqrt{3}}{3\pi } \frac{E[V]}{d}.
\label{eq:3411}
\end{equation}
The handover rate $H =  {E[N]}/{E[T_p]} $ is then given by
\begin{equation}
H = \frac{E[T]}{E[T]+E[S]} \cdot \frac{4 \sqrt{3}}{3\pi } \frac{E[V]}{d}.
\end{equation}
\label{pro:77}
\end{pro}
\begin{proof}
See Appendix \ref{proof:pro77}.
\end{proof}

We remark that the insights obtained by either approximate or exact approach are the same. Let us consider the simplified RWP mobility model where the mobile nodes do not have pause time and constant velocity, i.e., $V \equiv \nu$, the asymptotic handover rate in Corollary \ref{cro:2} can be further simplified as
\begin{equation}
{H}_{\textrm{app}} \sim \sqrt{\frac{\pi}{6\sqrt{3}}} \frac{\nu}{d},
\label{eq:341}
\end{equation}
and the exact handover rate $H$ is given by
\begin{equation}
H = \frac{ 4  \sqrt{3} }{ 3 \pi } \frac{\nu}{d },
\end{equation}
which is consistent with the one given in \cite{casares2011mobility}.
%Consider the first asymptotic case in Corollary \ref{cro:2} which is of particular interest and we examine in the sequel. 
Note that the hexagonal cell size $s_H$ is given by ${3\sqrt{3} d^2}/{2}$ in hexagonal tiling.  So (\ref{eq:341}) can be written as 
$
{H}_{\textrm{app}} \sim \sqrt{\frac{\pi}{4}} \frac{ \nu} {\sqrt{s_H}}
$.
Similarly, the exact method yields $ H =  \frac{4}{\pi} \sqrt{\frac{\sqrt{3}}{2}} \frac{\nu}{{\sqrt{s_H}}} $.
Both results imply that handover rate is inversely proportional to the square root of cell size $s_H$. In other words, if we deploy more small cells and increase the BS density say by 4 times in current cellular network, we would expect the handover rate to be roughly increased by 2 times. Interestingly, the mobility parameter $\lambda$ does not play a role, while the velocity and the cell size affect the handover rate in a trade-off manner. Fig. \ref{fig:44} compares the number of handovers obtained by simulation to the counterparts evaluated by analytic formula (\ref{eq:27}) and (\ref{eq:3411}) respectively. It is shown that the exact analytic result closely matches the simulation while the approximation approach tends to underestimate the real number of handovers.

\begin{figure}
\centering
\includegraphics[width=8cm]{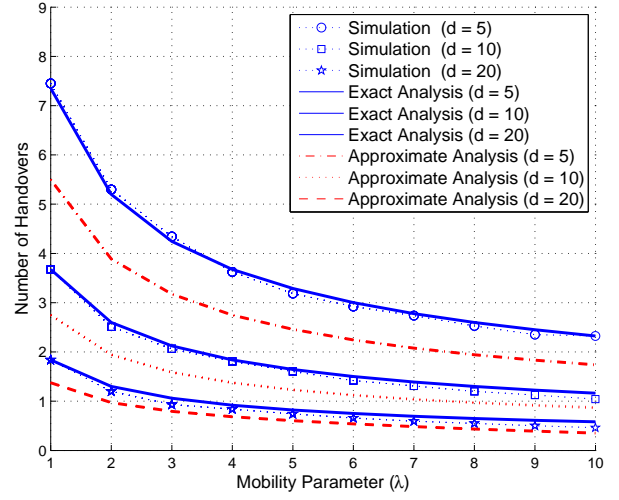}
\caption{Handover in hexagonal model  by analysis and simulation with velocity $\nu \equiv 1$ and no pause time. }
\label{fig:44}
\end{figure}

\subsection{Sojourn Time}

Sojourn time is defined as the expected duration that the mobile node stays within a particular serving cell. For brevity, we only consider the simplified RWP mobility model where the mobile nodes do not have pause time and constant velocity, i.e., $V \equiv \nu$. Then sojourn time in any cell with coverage area $\mathcal{C}$ can be computed as 
\begin{align}
S_T = E[T] \cdot \int_{\mathcal{C}} P({\rm d}A(r,\theta)),
\label{eq:38}
\end{align}
where $P(dA(r,\theta))$ is the probability distribution of the spatial node density $f (r,\theta)$ given in Theorem \ref{pro:6}. For brevity, we only focus on the sojourn time  in the cell where the connection is initiated during one movement period in the sequel.  

Assuming the mobile node is co-located with its currently associated BS at the origin, the sojourn time can be computed as 
\begin{align}
S_T 
=E[T] \cdot \frac{4\sqrt{\lambda}}{\pi}  ( &\int_{0}^{d/2} \int_{0}^{\sqrt{3}d/2} f(x,y)  {\rm d}y  {\rm d}x \notag \\
&+ \int_{d/2}^{d} \int_{0}^{-\sqrt{3}x + \sqrt{3}d } f(x,y)  {\rm d}y  {\rm d}x  ),
\label{eq:39}
\end{align}
where $f(x,y) = { \exp(- \lambda \pi (x^2 + y^2)) }/{ \sqrt{x^2 + y^2}}$. However, no closed form result is available for (\ref{eq:39}), though it can be evaluated using numerical methods. We provide explicit formulas for the lower and upper bounds of (\ref{eq:39}) in Prop. \ref{pro:10}, whose proof is omitted for brevity.
\begin{pro}
The sojourn time $S_T$ can be bounded as $S_T^L \leq S_T \leq S_T^U$ where
\begin{align}
S_T^L \triangleq  E[T] \cdot (1- 2Q (\sqrt{\frac{3}{2}\pi \lambda} d )   ),  \\  
S_T^U \triangleq E[T] \cdot (1- 2Q (\sqrt{2\pi \lambda} d ) ).
\label{eq:40}
\end{align}
\label{pro:10}
\end{pro}

Assuming that $E[T]$ is finite\footnote{This is true for both the case with constant velocity and the case with uniformly distributed velocity on $[v_{\min}, v_{\max}]$, as shown in Section \ref{subsec:1}.}, $S_T$ tends to $0$ since both the lower bound $S_T^L$ and upper bound $S_T^U$ tend to $0$ as $d \to 0$. This is an intuitive result. So cellular networks with small cells need to be equipped with fast algorithms. Otherwise, other handover strategies may be needed if the sojourn time in a cell is shorter than the time needed to complete the handover procedure. The analysis of the impact of mobility parameter $\lambda$ is more involved. To obtain insight, consider the constant velocity case which yields the following result.

\begin{pro}
Assume that $V \equiv \nu$. Then as $\lambda \to 0$, sojourn time $S_T \sim \alpha {d}/{\nu}$, where $\alpha \in (\frac{\sqrt{3}}{2},1)$ is a constant.
\label{pro:11}
\end{pro}
\begin{proof}
See Appendix \ref{proof:pro11}.
\end{proof}

From Prop. \ref{pro:11}, we observe the interesting fact that sojourn time converges to $\alpha d/\nu$ as the mobility parameter $\lambda$ goes to zero. Then what matters is the velocity: Sojourn time is inversely proportional to the velocity. Also, we can express $S_T$ as
$S_T \sim \sqrt{ \frac{2}{3\sqrt{3}}  } \frac{1}{\nu} \cdot \sqrt{s_H}   $ where recall that $s_H$ is the hexagonal cell size. This shows that sojourn time is proportional to the square root of cell size, which is contrary to the asymptotic result for handover rate.

\section{Applications to Poisson-Voronoi Modeled Cellular Networks}
\label{sec:app2}

In this subsection, we apply the proposed RWP mobility model to analyze handover rate and sojourn time in cellular networks modeled by Poisson-Voronoi tessellation \cite{moller1994lectures, baccelli2009stochastic}. We  first give a brief description on Poisson-Voronoi tessellation. Consider a locally finite set $\phi = \{ \boldsymbol{x}_i \}$ of points $\boldsymbol{x}_i \in \mathbb{R}^2$, referred to as nuclei. The Voronoi cell $\mathcal{C}_{\boldsymbol{x}_i} (\phi)$ of point $\boldsymbol{x}_i$ with respect to $\phi$ is defined as
\begin{align}
\mathcal{C}_{\boldsymbol{x}_i} (\phi) = \{ y \in \mathbb{R}^2:  \parallel \boldsymbol{y} - \boldsymbol{x}_i \parallel_2 \ \leq \  \parallel \boldsymbol{y} - \boldsymbol{x}_j \parallel_2, \forall x_j \in \phi  \} .   \notag
\label{eq:44}
\end{align}
Let $\epsilon_{\boldsymbol{x}}$ be the Dirac measure at $\boldsymbol{x}$, i.e., for $A \in \mathbb{R}^2$, $\epsilon_{\boldsymbol{x}} (A) = 1$ if $\boldsymbol{x} \in A$, and $0$ otherwise.
Then the spatial point process $\Phi$ can be written as $\Phi = \sum_i \epsilon_{\boldsymbol{x}_i}$, and the Poisson-Voronoi tessellation is defined as follows \cite{moller1994lectures}.

\begin{defi}
For a spatial Poisson point process $\Phi = \sum_i \epsilon_{\boldsymbol{x}_i}$ on $\mathbb{R}^2$, the union of the associated Voronoi cells, i.e., $\mathcal{V} (\Phi) = \bigcup_{\boldsymbol{x}_i \in \Phi} \mathcal{C}_{\boldsymbol{x}_i} (\Phi)
$, is called Poisson-Voronoi tessellation.
\end{defi}

In cellular networks modeled by a Poisson-Voronoi tessellation, the BSs are the nuclei distributed according to some PPP $\Phi$ in $\mathbb{R}^2$. Besides, each BS $\boldsymbol{x}_i$ serves mobile users which are located within its Voronoi cell $\mathcal{C}_{\boldsymbol{x}_i} (\Phi)$. The latter assumption is equivalent to the hypothesis of the nearest BS association strategy. In the sequel, we also assume that the PPP $\Phi$ modeling the BS positions is homogeneous and its intensity is denoted by $\mu$.

\subsection{Handover Rate}

Assume that the mobile node is located at the origin and let $\boldsymbol{X}_0 $ and $\boldsymbol{X}_1$ be two successive waypoints. Conditioned on the position of $\boldsymbol{X}_1$ and a given realization of the Poisson-Voronoi tessellation, the number of handovers equals the number of intersections of the segment $[\boldsymbol{X}_0,\boldsymbol{X}_1]$ and the boundary of the Poisson-Voronoi tessellation. Then  we can obtain the expected number of handovers by averaging over the spatial distribution of $\boldsymbol{X}_1$ and the distribution of Poisson-Voronoi tessellation. This is the main idea used in proving Prop. \ref{pro:12}. 

\begin{pro}
Let $\mu$ be the intensity of the homogeneously PPP distributed BSs and $\lambda$ the mobility parameter. The expected number of handovers $E[N]$ during one movement period  is given by
\begin{equation}
\setlength{\belowdisplayskip}{3pt}
E [ N ] = \frac{2}{\pi} \sqrt{ \frac{\mu}{\lambda} }.
\label{eq:52}
\end{equation}
The handover rate $H =  {E[N]}/{E[T_p]} $ is then given by
\begin{equation}
H = \frac{1}{E[T]+E[S]} \frac{2}{\pi} \sqrt{ \frac{\mu}{\lambda} }.
\end{equation}
\label{pro:12}
\end{pro}
\begin{proof}
See Appendix \ref{proof:pro12}.
\end{proof}

If we assume no pause time and constant velocity $\nu$, $H$ in Prop. \ref{pro:12} can be simplified as 
\begin{equation}
H = \frac{4}{\pi} \nu \sqrt{ \mu },
\end{equation}
which is consistent with the one given in \cite{baccelli1999poisson}. In Poisson-Voronoi tessellation with the nuclei being homogeneous PPP $\Phi$ of intensity $\mu$, the expected value of the size $s_P$  of a typical Voronoi cell is given by $s_P = E[ | C_o (\Phi) |_2 ] = {1}/{\mu}$ \cite{stoyan1995stochastic},
where $C_o (\Phi)$ is the typical cell and $| C_o (\Phi) |_2$ denotes the area of $C_o (\Phi)$. So if we assume no pause time and constant velocity $\nu$, $H$ in Prop. \ref{pro:12} can be simplified as 
$
\displaystyle {H}= \frac{4}{\pi} \nu  /{\sqrt{s_P}}, 
$
which implies that the handover rate is inversely proportional to the square root of the cell size. This is consistent with the results in the hexagonal model.

Fig. \ref{fig:14} illustrates that the analytical result (\ref{eq:52}) matches the simulation result quite well. Also, we compare the handover rate of Poisson-Voronoi model, exact and approximate handover rate of hexagonal model in Fig. \ref{fig:201} as a function BS intensity. They all indicate that handover rate grows linearly with the square root of the BS's intensity $\sqrt{\mu}$. We further evaluate the three types of analytic results on handover, i.e., Prop. \ref{pro:7}, Prop. \ref{pro:77} and \ref{pro:12} by simulating the proposed RWP mobility model using the real-world data of macro-BS deployment in a cellular network, provided by a major service provider. Recall that we assume each BS serves the mobile users located within its Voronoi cell. A handover occurs when the mobile crosses the cell boundary. Thus, only the BS location data are relevant in the simulation. There are 400 BSs which are distributed in a relatively flat urban area. These 400 BSs roughly occupy a 105$\times$90 km area. We normalize the network size to be 1$\times$1. So $\mu = 400$ in the Poisson-Voronoi model, while the side length $d$ in hexagonal model is determined through $ 3 \sqrt{3} d^2 / 2 = 1/400$. The results are shown in Fig. \ref{fig:20}. It can be seen that the Poisson-Voronoi model is about as accurate as the hexagonal model in predicting the number of handovers. Meanwhile, the approximate analytic result underestimates the number of handovers. 

\begin{figure}
\centering
\includegraphics[width=8cm]{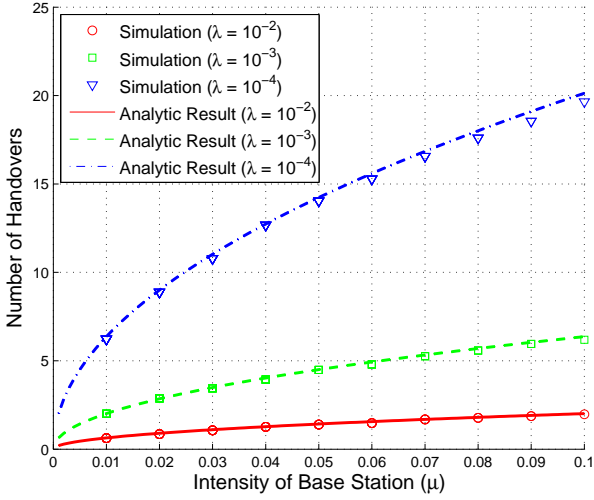}
\caption{Handover in Poisson-Voronoi model by analysis and simulation with velocity $\nu \equiv 1$ and no pause time. The number of handovers is proportional to the square root of the BS's intensity $\mu$, and is inversely proportional to mobility parameter $\lambda$.}
\label{fig:14}
\end{figure}

\begin{figure}
\centering
\includegraphics[width=8cm]{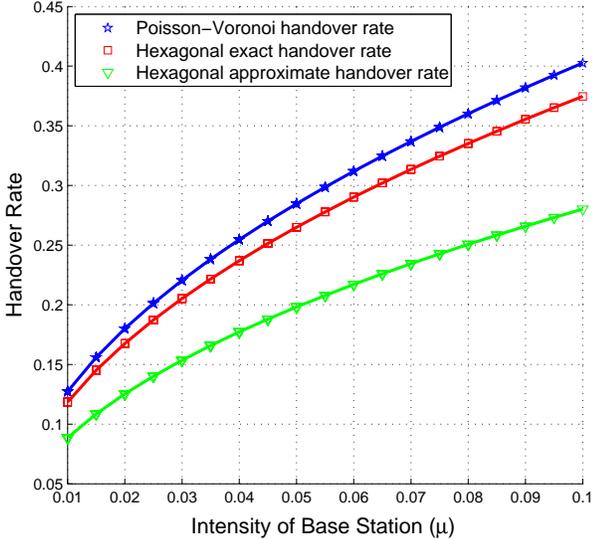}
\caption{Comparison of handover rate in Poisson-Voronoi model, exact and approximate handover rate in hexagonal model  with velocity $\nu \equiv 1$, normalized BS's intensity $\mu = 1$, mobility parameter $\lambda=1$ and no pause time. They all imply that the handover rate grows linearly with $\sqrt{\mu}$.}
\label{fig:201}
\end{figure}

\begin{figure}
\centering
\includegraphics[width=8cm]{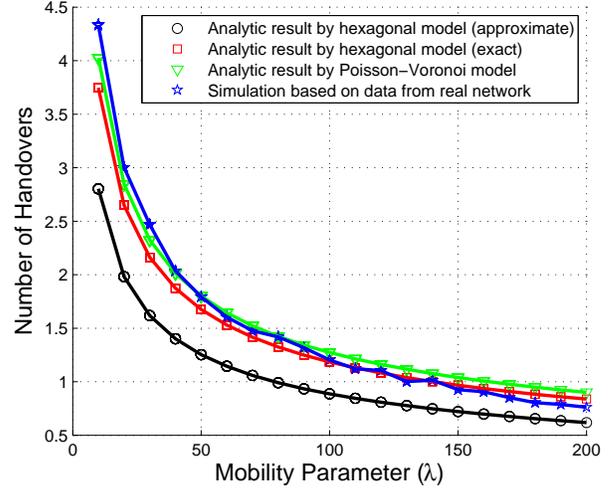}
\caption{Evaluation of handover by an actual macro-BS deployment with velocity $\nu \equiv 1$, no pause time and normalized BS's intensity $\mu = 400$. Poisson-Voronoi model is about as accurate in terms of handover evaluation as hexagonal model.}
\label{fig:20}
\end{figure}

\subsection{Sojourn Time}

Unlike the hexagonal model, (\ref{eq:38}) is not sufficient for computing the sojourn time in Poisson-Voronoi tessellation modeled cellular networks. Indeed, assuming that the mobile user is located at the origin, the sojourn time in the cell involves an additional source of randomness - the Poisson-Voronoi tessellation. So even averaging over the spatial node distribution sojourn time is still a random variable. In the sequel, we aim to characterize the probability distribution of this sojourn time. To this end, we first introduce the concept of contact distribution. Consider a random closed set $\mathcal{Z}$ and a convex compact test set $\mathcal{B}$ containing the origin $o$. Then the  contact distribution is defined as \cite{stoyan1995stochastic}
\begin{align}
H_{\mathcal{B}} (r) = P(  \mathcal{Z} \cap r \mathcal{B} \neq \emptyset |  o \notin \mathcal{Z}), r \geq 0.
\end{align}
In this paper, we are particularly interested in a special case: the linear contact distribution function $H_{l} (r)$ with the test set $\mathcal{B}$ being a segment of unit length $l$. Since Poisson-Voronoi tessellation is isotropic, the orientation of this segment $l$ is not important. For Poisson-Voronoi tessellation modeled cellular network, the random closed set $\mathcal{Z}$ of interest is the union of all cell boundaries $\cup_{\boldsymbol{x}_i \in \Phi} \partial \mathcal{C}_{\boldsymbol{x}_i} (\Phi) $, where $\partial \mathcal{C}_{\boldsymbol{x}_i} (\Phi)$ denotes the boundary of the cell $\mathcal{C}_{\boldsymbol{x}_i} (\Phi)$.

With a slight abuse of notation, let $\mathcal{C}_{o} (\Phi)$ denote the Voronoi cell containing the origin $o$. We can further simplify $H_{l} (r)$ as follows:
\begin{align}
H_{l} (r) = 1 - \frac{ P ( \mathcal{Z} \cap r l = \emptyset ) }{ 1 - P( o \in \mathcal{Z} )  } =  1 - P(  r l \subseteq \mathcal{C}_{o} (\Phi) ).
\label{eq:55}
\end{align}
The last equality follows because 1) the origin $o$ is contained in the interior of $\mathcal{C}_{o} (\Phi)$ almost surely and thus $P( o \in \mathcal{Z} ) = 0 $ where $\mathcal{Z} = \cup_{\boldsymbol{x}_i \in \Phi} \partial \mathcal{C}_{\boldsymbol{x}_i} (\Phi) $, and 2) the event $\{\mathcal{Z} \cap r l = \emptyset\}$ is equivalent to that $\{r l \subseteq \mathcal{C}_{o} (\Phi)\}$. Now we are in a position to characterize the sojourn time $S_T$.
\begin{pro}
Let $\mu$ be the intensity of the homogeneously PPP distributed BSs and $\lambda$ the mobility parameter. The pdf of sojourn time $S_T$ is given by
\begin{align}
f_{S_T}(t) &=\frac{1 }{2 \sqrt{\lambda} E[T]} \exp\left( \frac{1}{2} ( Q^{(-1)} ( \frac{1}{2}(1 - \frac{t}{E[T]}) ) )^2 \right)  h_l \left(  x \right),
\label{eq:577}
\end{align}
where $E[T]$ is the expected transition time, $x=\frac{1}{\sqrt{2\pi \lambda}}  Q^{(-1)} ( \frac{1}{2}(1 - \frac{t}{E[T]}) )$ and $h_l (r)$ is given by
\begin{align}
h_l (r) = 4 \pi \mu^2 \int_0^\pi & \int_0^{\pi - \alpha} r^3 \frac{  \sin^2\alpha \sin \beta }{ \sin^4 (\alpha + \beta) } b_0(\beta) \notag \\
& \exp( - \mu V_2 (r,\alpha, \beta) )  {\rm d}\beta  {\rm d}\alpha , r \geq 0,
\label{eq:56}
\end{align}
where 
$
V_2 (r,\alpha, \beta) = \pi   \rho^2 ( a_0 (\alpha) + a_1 (\alpha)  ) + \pi ( r^2 + \rho^2 - 2r\rho \cos \alpha ) (  a_0 ( \beta   ) + a_1 (  \beta  ) )  
$
with $\rho = \frac{r \sin \beta}{\sin(\alpha + \beta)}, a_0 (\theta) = 1 - \frac{\theta}{\pi}$, and $a_1 (\theta) = \frac{\sin 2 \theta}{2\pi}$, and $b_0 (\beta) = \frac{  (\pi-\beta) \cos \beta  + \sin \beta }{ \pi }$.
\label{pro:13}
\end{pro}
\begin{proof}
By Theorem $\ref{pro:6}$, we have
\begin{align}
S_T (r) &= E[T] \int_0^{2\pi} \int_0^{r} \frac{\sqrt{\lambda}}{\pi x} \exp (-\lambda \pi x^2 )   x  {\rm d}x  {\rm d} \alpha \notag \\
&= E[T] \cdot (1 - 2 Q ( \sqrt{2\pi \lambda} r  ) ). \notag
\end{align}
Clearly, $F_{S_T}(t) = 0$ if $t<0$, and $F_{S_T}(t) = 1$ if $t\geq E[T]$. If $0 \leq t < E[T]$,
\begin{align}
F_{S_T}(t) &= P (S_T \leq t ) = P (E[T] \cdot (1 - 2 Q ( \sqrt{2\pi \lambda} r  ) ) \leq t ) \notag \\
&= P \left( 0 \leq r \leq \frac{1}{\sqrt{2\pi \lambda}}  Q^{(-1)} ( \frac{1}{2}(1 - \frac{t}{E[T]}) )  \right)  \notag \\
&= H_l \left( \frac{1}{\sqrt{2\pi \lambda}}  Q^{(-1)} ( \frac{1}{2}(1 - \frac{t}{E[T]}) ) \right)  \notag \\
&= \int^{\frac{1}{\sqrt{2\pi \lambda}}  Q^{(-1)} ( \frac{1}{2}(1 - \frac{t}{E[T]}) )}_{-\infty} h_l (r)  {\rm d}r. \notag
\end{align}
Taking the derivative with respect to $t$ for both sides yields (\ref{eq:577}). The closed form expression  of $h_l (r)$ in (\ref{eq:56}) has been given in \cite{gilbert1962random}. This completes the proof.
\end{proof}

Using the previous results for expected transition time, the pdf of the sojourn time under different velocity distributions can be derived from Prop. \ref{pro:13}. Note that the distribution of the sojourn time  is instrumental in studying the ping-pong behavior for mobility enhancement in cellular networks \cite{3gppMobility}. Though a closed form expression is not available, the pdf of the sojourn time given in Proposition \ref{pro:13} only involves a double integral, which is tractable for efficient numerical evaluation. We plot the pdf of sojourn time in Fig. \ref{fig:11} for illustration. As expected, the smaller $\mu$ is, the more likely the mobile node stays longer within the cell. This is intuitive since a smaller $\mu$ implies larger cell sizes on average. As a result, it is less likely that the mobile node would move out of the cell.

\begin{figure}
\centering
\includegraphics[width=8cm]{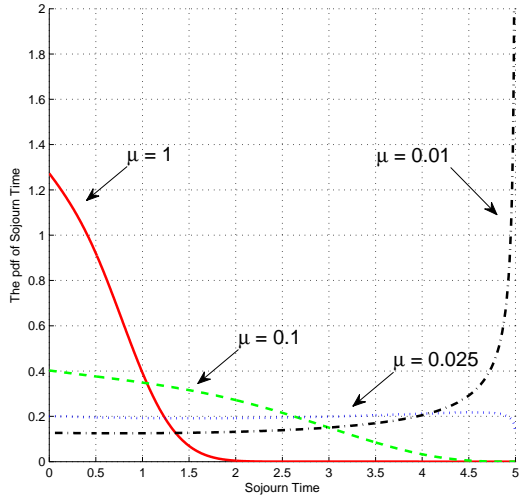}
\caption{Plot of pdf of sojourn time in Poisson-Voronoi model with mobility parameter $\lambda= 0.01$ and velocity $\nu \equiv 1$.}
\label{fig:11}
\end{figure}

We also compare the analytic result about sojourn time in the Poisson-Voronoi model to its deterministic counterpart in the hexagonal model in Fig. \ref{fig:21}. It is shown that the analytic result in the Poisson-Voronoi model is more conservative and yields smaller mean sojourn times than its counterpart in the hexagonal model. Besides, we can see that the upper and lower bounds of the sojourn time in hexagonal model are pretty tight.

\begin{figure}
\centering
\includegraphics[width=8cm]{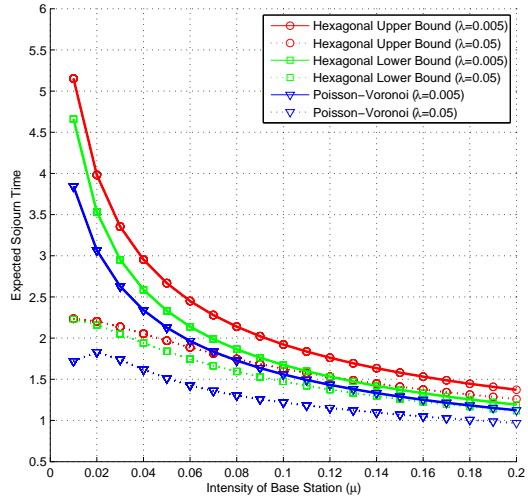}
\caption{Comparison of sojourn time in hexagonal tiling and Poisson-Voronoi tessellation with velocity $\nu \equiv 1$. Hexagonal model yields larger sojourn time than that of Poisson-Voronoi model. Besides, the upper and lower bounds for hexagonal model are tight.}
\label{fig:21}
\end{figure}

\section{Conclusions and Future Work}

In this paper, we study the critical mobility issue in cellular networks. To this end, we first propose a tractable RWP mobility model defined on the entire plane. The various properties of the mobility model are carefully studied and simple analytical expressions are obtained. Then we utilize this tractable mobility model to analyze the handover rate and sojourn time in cellular networks. The analysis is carried out for cellular networks under both hexagonal and Poisson-Voronoi models. We derive closed form expressions and/or bounds for the performance metrics in question. These analytical results are instrumental for mobility management in cellular networks.

Note that the proposed RWP mobility model represents the real movement of mobile nodes in a simplified manner. Thus, it does not capture some other mobility characteristics such as temporal and spatial dependency of the mobility pattern. It is desirable to extend the current model further to incorporate these extra mobility characteristics while maintaining a certain degree of tractability. It would be rather interesting (and challenging) to characterize the intuitive tradeoff between the complexity and tractability of the mobility models. Besides, it is also interesting to evaluate the performance of the many wireless protocols by applying the proposed tractable model in theory and/or simulation.

\section*{Acknowledgment}

The authors would like to thank the anonymous reviewers
for their valuable comments and suggestions, which helped the authors significantly improve the quality of the paper.

\appendix

\subsection{Proof of Theorem \ref{pro:6}}
\label{proof:pro6}

We first derive the pdf of the random waypoint $\boldsymbol{X}_1 = (R_1, \Theta_1)$ as follows:
\begin{align}
f_{\boldsymbol{X}_1}(r,\theta) &= \lim_{\Delta r \to 0} \frac{P(R_1 \leq r + \Delta r) - P(R_1 \leq r) }{  \int_{0}^{2 \pi} \int_{r}^{r + \Delta r}  x  {\rm d}x    {\rm d} \phi }  \notag \\
 &= \lim_{\Delta r \to 0} \frac{ \exp(- \lambda \pi r^2)  -\exp(- \lambda \pi (r + \Delta r)^2) }{2  \pi r \Delta r } \notag  \\
&= \lim_{\Delta r \to 0} \frac{ 2 \pi \lambda (r + \Delta r) \exp(- \lambda \pi (r + \Delta r)^2) }{2  \pi r  }  \notag \\
&= \lambda \exp(-\lambda_i \pi r^2). \notag
\end{align}

This result is summarized in Lemma \ref{lem:1} for ease of reference.
\begin{lem}
Given $\boldsymbol{X}_0$ is at the origin, the pdf of the random waypoint $\boldsymbol{X}_1 = (R_1, \Theta_1)$ is given by
\begin{equation}
f_{\boldsymbol{X}_1}(r,\theta) = \lambda \exp(-\lambda \pi r^2).
\end{equation}
\label{lem:1}
\end{lem}

We next derive the pdf of the spatial node distribution. The main technique of the following proof is inspired by \cite{bettstetter2003node}, which is also adopted in \cite{hyytia2006spatial}. Consider a small set $dA$ located at $(r,\theta)$. Let $\vec{L}_1$  denote the vector $\boldsymbol{X}_1 - \boldsymbol{X}_0$ and $|\vec{L}_1| = L_1$. Then the pdf of the spatial node distribution can be interpreted as the ratio of the expected proportion of transition time in the set ${\rm d}A$ to the area $|{\rm d}A|$, i.e.,
\begin{align}
f(r, \theta) = \frac{ E [| \vec{L}_1 \cap {\rm d}A |  / V] }{ E[L_1/ V  ] |{\rm d}A| } = \frac{ E [| \vec{L}_1 \cap {\rm d}A |  ] }{ E[L_1  ] |{\rm d}A| },
\label{eq:apro63}
\end{align}
where the second equality follows from the independence of $V$ and the waypoints. Note that 
\begin{align}
&\frac{ E [| \vec{L}_1 \cap dA |  ] }{|dA|} = \frac{ \int_{r}^{\infty} f_{\boldsymbol{X}_1}(x,\theta)  x  {\rm d}x {\rm d} \theta \cdot \Delta l }{ r \cdot {\rm d} \theta \cdot \Delta l }  \notag \\
&=  \frac{ \int_{r}^{\infty} \lambda \exp(-\lambda \pi x^2)  x  {\rm d}x } {r} =\frac{1}{2 \pi r} \exp(-\lambda \pi r^2),
\label{eq:apro64}
\end{align}
where $\Delta l$ denotes the length of the small intersection if $\vec{L}_1$ intersects ${\rm d}A$, and we apply Lemma \ref{lem:1} in the second equality in (\ref{eq:apro64}). The first equality in (\ref{eq:apro64}) can be explained through Fig. \ref{fig:1} as follows. The intersection of $| \vec{L}_1 \cap {\rm d}A |$ is $\Delta l$ if $\boldsymbol{X}_1$ is in the shaded area and $0$ otherwise. Noting  that the probability of the event that $\boldsymbol{X}_1$ is in the shaded area is $\int_{r}^{\infty} f_{\boldsymbol{X}_1}(x,\theta)  x  {\rm d}x  {\rm d}\theta$, we have
\begin{align}
E [| \vec{L}_1 \cap dA |  ] =  \Delta l \cdot \int_{r}^{\infty} f_{\boldsymbol{X}_1}(x,\theta)  x  {\rm d}x  {\rm d}\theta.
\end{align}
This and $|dA| = r \cdot d \theta \cdot \Delta l$ give the desired equality.  So
\begin{align}
f(r, \theta) = \frac{\exp(-\lambda \pi r^2)}{ 2 \pi r E[L_1]  } = \frac{\sqrt{\lambda} \exp(-\lambda \pi r^2)}{ \pi r  },
\label{eq:apro65} 
\end{align}
where we use the result that $E[L_1] = 1/2\sqrt{\lambda}$ in the last equality. This completes the proof.

\subsection{Proof of Proposition \ref{pro:7}}
\label{proof:lem2}

From Lemma \ref{lem:1}, $f_{\boldsymbol{X}_1}(r,\theta) = \lambda \exp(-\lambda \pi r^2)$. We compute $E[N]_{\textrm{app}}$ as follows.
\begin{align}
&E[N]_{\textrm{app}} = \sum_{n = 1}^{\infty} n \int_{0}^{2\pi} \int_{(2n-1) R}^{(2n+1)R} \lambda \exp(-\lambda \pi r^2) r  {\rm d}r  {\rm d}\theta \notag \\
&= \sum_{n = 1}^{\infty} n \left( \exp ( - \pi (2n-1)^2  R^2 \lambda) - \exp ( - \pi (2n+1)^2  R^2 \lambda) \right) \notag \\
&= \sum_{n = 0}^{\infty} \exp ( - \pi (2n+1)^2  R^2 \lambda)   \notag 
\end{align}
Substituting $R = \displaystyle \sqrt { \frac{C}{\pi} } =\displaystyle \frac{\sqrt[4]{27} d }{\sqrt{2 \pi}}$, we obtain
\begin{align}
&E[N]_{\textrm{app}} = \sum_{n = 0}^{\infty} \exp \left( - \frac{3\sqrt{3}}{2} (2n+1)^2  d^2 \lambda \right). \notag
\label{eq:30}
\end{align}
Then the lower bound $N^L_{\textrm{approx}}$ is derived as follows.
\begin{align}
E[N]_{\textrm{app}} 
&\geq \int_{0}^{\infty} \exp \left( - \frac{3\sqrt{3}}{2} (2x+1)^2  d^2 \lambda \right)  {\rm d}x  \notag \\
&= \sqrt{\frac{\pi}{6\sqrt{3} \lambda d^2}} Q \left( \sqrt{3 \sqrt{3} \lambda d^2 } \right) \triangleq E[N]^L_{\textrm{app}}. \notag
\end{align}
The upper bound $E[N]^U_{\textrm{app}}$ can be derived in a similar manner. For the remaining proof, denote $t = \sqrt{3\sqrt{3}\lambda d^2}$. Note that $t > 0$. Then the difference $\triangle N_{\textrm{app}} =  E[N]^U_{\textrm{app}} - E[N]^L_{\textrm{app}}$ can be compactly written as 
$ \sqrt{\frac{\pi}{2}} \frac{1}{t} ( 1- 2 Q(t) ) $ which we denote by $g(t)$. We claim $g(t)$ is strictly decreasing when $t > 0$. To see this, define
$
h(t) = \sqrt{\frac{2}{\pi}} t \exp( - \frac{t^2}{2})  + 2 Q(t) - 1
$
whose derivative is ``$- \sqrt{\frac{2}{\pi}} t^2 \exp( - \frac{t^2}{2}) $'' which is strictly negative when $t > 0$, implying $h(t)$ is strictly decreasing when $t > 0$. Meanwhile, $h(0) = 0$. So $h(t) < 0$ when $t > 0$. Now it is clear that the derivative of $g(t)$ given by $\displaystyle g^{(1)} (t) = \sqrt{\frac{\pi}{2}} \frac{h(t)}{t^2}$ is strictly negative when $t > 0$.  Thus, $g(t)$ is strictly decreasing when $t > 0$. Besides, 
\begin{align}
&\lim_{t \to 0} \sqrt{\frac{\pi}{2}} \frac{1}{t} ( 1- 2 Q(t) ) = \lim_{t \to 0} \exp(- \frac{t^2}{2}) = 1, \notag \\
&\lim_{t \to \infty} \sqrt{\frac{\pi}{2}} \frac{1}{t} ( 1- 2 Q(t) ) = 0. \notag
\end{align}
The desired results follow by further observing $t$ is a strictly increasing continuous function of $\lambda d^2$.

\subsection{Proof of Proposition \ref{pro:77}}
\label{proof:pro77}

We use (generalized) argument for Buffon's needle problem in this proof. Consider the typical node located in area $\mathcal{A}$ of size $|\mathcal{A}|$ and cell boundaries $\mathcal{B}_A$ of length $|\mathcal{B}_A|$. Then the probability that this node crosses the small boundary $\Delta b$ within a small time interval $\Delta t$ is 
$
\Delta p =  \frac{1}{|\mathcal{A}|}  \cdot \Delta b \cdot \nu \Delta t E [|\sin \Theta|] 
$,
where $\Theta$ is uniformly distributed on $[0,2\pi]$ (following from our mobility model) and thus $E [|\sin \Theta|] = \frac{2}{\pi}$. So the probability that this node crosses $\mathcal{B}_A$ is
$
p ( \mathcal{A} ) = \Delta p \cdot \frac{ |\mathcal{B}| }{\Delta b} = \frac{2}{\pi} \frac{|\mathcal{B}_A|}{|\mathcal{A}|}  \cdot   \nu \Delta t 
$. Thus, conditioned on moving, the handover rate 
\begin{align}
H^\star &= \lim_{ |\mathcal{A}| \to \infty }  \frac{ p ( \mathcal{A} ) }{ \Delta t } = \frac{2}{\pi}   \nu  \cdot   \lim_{ |\mathcal{A}| \to \infty } \frac{|\mathcal{B}_A|}{|\mathcal{A}|} \notag \\
&= \frac{2}{\pi}   \nu  \cdot \frac{ 9 d }{ 9 \sqrt{3}  d^2 /2  } = \frac{4 \sqrt{3}}{3\pi } \frac{\nu}{d}. \notag
\end{align}
Correspondingly, $ \displaystyle E[N] = E[H^\star] \cdot E[T] = \frac{4 \sqrt{3}}{3\pi } \frac{E[V]}{d} \cdot E[T] $ and $\displaystyle H = \frac{E[T]}{E[T]+E[S]} \frac{4 \sqrt{3}}{3\pi } \frac{E[V]}{d}$.

\subsection{Proof of Proposition \ref{pro:11}}
\label{proof:pro11}

\begin{figure}
\centering
\includegraphics[width=6cm]{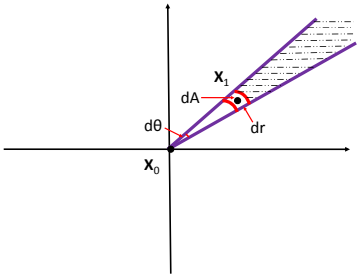}
\caption{A geometric illustration for the proof of Prop. \ref{pro:6}}
\label{fig:1}
\end{figure}

Applying the general expression (\ref{eq:39}) for $S_T$ yields
\begin{align}
 S_T=& \frac{4}{\pi \nu} f(\xi_x d,\xi_y d) \cdot \notag \\
 & \int_{0}^{d/2} \int_{0}^{\sqrt{3}d/2}  {\rm d}y  {\rm d}x + \int_{d/2}^{d} \int_{0}^{-\sqrt{3}x + \sqrt{3}d }   {\rm d}y  {\rm d}x,   \notag 
\end{align}
which equals
\begin{align}
= & \frac{3\sqrt{3}}{2 \pi} \frac{ \exp(- \lambda \pi (\xi_x^2 d^2 + \xi_x^2 d^2)) }{ \sqrt{\xi_x^2 + \xi_x^2}} \frac{d}{\nu},
\end{align}
where $(\xi_x,\xi_y)$ are a pair of constants in the region $A:= \{ (x,y):  x \in [0,1], y \in [0,\frac{\sqrt{3}}{2}] \textrm{ if } x \in [0,\frac{1}{2}) \textrm{ and } y \in [0,\sqrt{3} (1-\xi_x) ] \textrm{ if } x \in [\frac{1}{2},1] \} $, and  in the last equality we use $E[T] = \frac{1}{2 \nu \sqrt{\lambda}}$ derived for constant velocity case in Section \ref{sec:stochastic}, and mean value theorem for integrals. Then, 
$
\lim_{\lambda \to 0}  S_T   = \frac{3\sqrt{3}}{2 \pi} \frac{1 }{ \sqrt{\xi_x^2 + \xi_x^2}} \frac{d}{\nu}.
$
Next we give bounds on the constant $\alpha := \frac{3\sqrt{3}}{2 \pi} \frac{1 }{ \sqrt{\xi_x^2 + \xi_x^2}} $. Using the bounds given in Prop. \ref{pro:10},
\begin{align}
\lim_{\lambda \to 0} S^L_T &= \lim_{\lambda \to 0} \frac{1- 2Q \left(\sqrt{\frac{3}{2}\pi \lambda} d \right)}{2\nu \sqrt{\lambda}} \notag \\
&= \lim_{\lambda \to 0} \frac{\sqrt{3}}{2} \frac{d}{\nu} \exp(- \frac{3}{4}\pi d^2 \lambda) = \frac{\sqrt{3}}{2} \frac{d}{\nu} .
\label{eq:401}
\end{align}
Similarly, $\lim_{\lambda \to 0} S^U_T = {d}/{\nu} $. Since these bounds are strict, we conclude that $\alpha \in (\frac{\sqrt{3}}{2},1)$.

\subsection{Proof of Proposition \ref{pro:12}}
\label{proof:pro12}

This proposition can be proved by following the proof of Prop. \ref{pro:77}. Nevertheless, we provide an alternative proof here. To this end, we first introduce some terminologies for ease of exposition.
Consider the set $\boldsymbol{x}^{(k)}$ consisting of arbitrary $k$ distinct points in $\Phi$ on $\mathbb{R}^2$. Without loss of generality, we assume $\boldsymbol{x}^{(k)} = \{\boldsymbol{x}_1,...,\boldsymbol{x}_k \}$. Then if the intersection 
$
\mathcal{F} ( \boldsymbol{x}^{(k)} | \Phi ) = \cap_{i=1}^k \mathcal{C}_{\boldsymbol{x}_i} (\Phi) \neq \emptyset,
$
then $\mathcal{F} ( \boldsymbol{x}^{(k)} | \Phi )$ is called a ($3-k$)-facet. Now let $\Phi_k$ be the set of all configurations $\boldsymbol{x}^{(3-k)} \subseteq \Phi$. That is, the intersection of the Voronoi cells associated with any configuration $\boldsymbol{x}^{(3-k)} \in \Phi_k$ is a $k$-facet.  For each configuration $\boldsymbol{x}^{(3-k)} = \{\boldsymbol{x}_1,...,\boldsymbol{x}_{3-k} \} \in \Phi_k$, we associate a point $\boldsymbol{z}_{  \boldsymbol{x}^{(3-k)} } (\Phi)$ called ``centroid'' such that for all $\boldsymbol{y} \in \mathbb{R}^2$,
$
\boldsymbol{z} (  \boldsymbol{x}^{(3-k)} + \boldsymbol{y} | \Phi + \boldsymbol{y}) = \boldsymbol{y} + \boldsymbol{z} (  \boldsymbol{x}^{(3-k)} | \Phi  )
$.
Note that there are several degrees of freedom in choosing the centroids. We refer to \cite{moller1994lectures} for more details.  

Consider the intersection of the Poisson-Voronoi tessellation with a fixed line $L$. Without loss of generality, we assume that $L$ contains the origin. Then the nonempty sectional cells 
$
\mathcal{\bar{C}}_{\boldsymbol{x}_i} (\Phi) = \mathcal{C}_{\boldsymbol{x}_i} (\Phi) \cap L
$
satisfy that (i) $L = \cup_{\boldsymbol{x}_i \in \Phi} \mathcal{\bar{C}}_{\boldsymbol{x}_i} (\Phi)$, and (ii) $ ri ( \mathcal{\bar{C}}_{\boldsymbol{x}_i} (\Phi) ) \cap ri ( \mathcal{\bar{C}}_{\boldsymbol{x}_j} (\Phi) ) = \emptyset, \forall i\neq j$ where $ri(\cdot)$ denotes the relative interior \cite{moller1994lectures}. So the sectional cells  $\mathcal{\bar{C}}_{\boldsymbol{x}_i} (\Phi) $ constitute a tessellation of $L$, denoted by $\mathcal{V}_L (\Phi)$. Now consider the intersection $\mathcal{F}_L ( \boldsymbol{x}^{(2)} | \Phi ) = \mathcal{F} ( \boldsymbol{x}^{(2)} | \Phi ) \cap L $ which can be either empty or a singleton. Clearly, $\mathcal{F}_L ( \boldsymbol{x}^{(2)} | \Phi )$ is the $0$-facet of tessellation $\mathcal{V}_L (\Phi)$ if $\mathcal{F}_L ( \boldsymbol{x}^{(2)} | \Phi )$ is nonempty. Now let $\Phi^L_{0} = \{ \boldsymbol{x}^{(2)} =\{\boldsymbol{x}_1, \boldsymbol{x}_2 \} \in \Phi_1: \mathcal{F}_L ( \boldsymbol{x}^{(2)} | \Phi ) \neq \emptyset \} $ which parametrizes the $0$-facets of $\mathcal{V}_L (\Phi)$. Then the intensity of the $0$-facets of $\mathcal{V}_L (\Phi)$ is well defined as
\begin{align}
\mu^L_0 = \frac{E[ \sum_{ \boldsymbol{x}^{(2)} \in \Phi^L_{0} } \chi \{ \boldsymbol{z}_L(\boldsymbol{x}^{(2)} | \Phi) \in \mathcal{B}_L  \}  ]}{ |\mathcal{B}_L|_1 },
\label{eq:48}
\end{align}
where $\chi\{\cdot \}$ is indicator function taking value $1$ if the event in its argument is true and $0$ otherwise, $\boldsymbol{z}_L(\boldsymbol{x}^{(2)} | \Phi)$ is the centroid of the $0$-facets of $\mathcal{V}_L (\Phi)$, $\mathcal{B}_L \subseteq L $ is any arbitrary Borel set and $|\mathcal{B}_L|_1$ is the volume of $\mathcal{B}_L$ with respect to $\mathbb{R}$. 

Without loss of generality, assume that $\boldsymbol{X}_0$ is at the origin. Let $L:=L(\boldsymbol{X}_0, \boldsymbol{X}_1)$ be the line containing $\boldsymbol{X}_0$ and $\boldsymbol{X}_1$, and $\mathcal{B}_L $ the interval $[\boldsymbol{X}_0, \boldsymbol{X}_1]$. Then conditioned on the next waypoint $\boldsymbol{X}_1 = (r,\theta)$, the expected number of handovers can be computed as
\begin{align}
E[N| \boldsymbol{X}_1 = (r,\theta)  ] &= E[ \sum_{ \boldsymbol{x} \in \cup_{\boldsymbol{x}_i \in \Phi} \partial \mathcal{C}_{\boldsymbol{x}_i} (\Phi)  \cap L } \chi \{ \boldsymbol{x} \in \mathcal{B}_L  \}  ] \notag \\
&= E[ \sum_{ \boldsymbol{x}^{(2)} \in \Phi^L_{0} } \chi \{ \boldsymbol{z}_L(\boldsymbol{x}^{(2)} | \Phi) \in \mathcal{B}_L  \}  ]  \label{eq:46} \\
&= \mu^L_0 |\boldsymbol{X}_0 - \boldsymbol{X}_1|_1 = \frac{4 r \sqrt{\mu} }{\pi}, \label{eq:47}
\end{align}
where $\partial \mathcal{C}_{\boldsymbol{x}_i} (\Phi) $ in (\ref{eq:46}) denotes the boundary of the cell $\mathcal{C}_{\boldsymbol{x}_i} (\Phi)$. The equality in (\ref{eq:46}) follows by choosing the centroids $\boldsymbol{z}_L(\boldsymbol{x}^{(2)}$ as follows: for any $\boldsymbol{x}^{(2)} \in \Phi^L_{0}$, choose the singleton $\mathcal{F} ( \boldsymbol{x}^{(2)} | \Phi ) \cap L$ as the centroid $\boldsymbol{z}_L(\boldsymbol{x}^{(2)} | \Phi)$ of the $0$-facets of $\mathcal{V}_L (\Phi)$. The equality in (\ref{eq:47}) follows from (\ref{eq:48}). The last equality follows since $ \mu^L_0 = {4 \sqrt{\mu} }/{\pi} $ for Poisson-Voronoi tessellation with intensity $\mu$ \cite{moller1994lectures}. Then handover rate can be computed as follows:
\begin{align}
H &= \frac{E[N]}{E[T]} = \frac{ E[ E[N | \boldsymbol{X}_1 = (r,\theta)  ] ]}{E[T]}  \notag \\
&= \frac{ 1 }{E[T]} \int_{0}^{2\pi} \int_0^{\infty} E[N | \boldsymbol{X}_1 = (r,\theta)  ] f_{\boldsymbol{X}_1}(r,\theta) r  {\rm d} r  {\rm d} \theta  \notag \\
&= \frac{ 1 }{E[T]} \int_{0}^{2\pi} \int_0^{\infty} \frac{4 r \sqrt{\mu} }{\pi}   \lambda \exp(-\lambda \pi r^2) r  {\rm d} r  {\rm d} \theta \notag \\
&= \frac{ 1 }{E[T]} \frac{4 \Gamma (\frac{3}{2})}{\pi^{\frac{3}{2}}} \sqrt{ \frac{\mu}{\lambda} } \label{eq:50} \\
&= \frac{ 1 }{E[T]} \frac{2}{\pi} \sqrt{ \frac{\mu}{\lambda} }.   \label{eq:51}
\end{align}
where in (\ref{eq:50}) we use Lemma \ref{lem:1}, and in (\ref{eq:51}) we apply the result that $
\int_{0}^{\infty} r^{\alpha - 1 } \exp(-\gamma \pi r^\beta) \ {\rm d} r = \frac{\Gamma (\alpha / \beta)}{\beta \gamma^{\alpha / \beta} } $ for $\alpha, \beta, \gamma > 0$. Plugging $E[T]$ yields the desired results.

\bibliographystyle{IEEEtran}
\bibliography{IEEEabrv,Reference}

\end{document}